\documentclass[11pt,twoside]{amsart}
\usepackage{amsmath, amsthm, amscd, amsfonts, amssymb, graphicx, color}
\usepackage[bookmarksnumbered, plainpages]{hyperref}
\newtheorem{thm}{Theorem}[section]
\newtheorem{cor}[thm]{Corollary}
\newtheorem{lem}[thm]{Lemma}
\newtheorem{prop}[thm]{Proposition}
\newtheorem{defn}[thm]{Definition}
\newtheorem{rem}[thm]{\bf{Remark}}

\numberwithin{equation}{section}
\def\pn{\par\noindent}

\newcommand{\ch}{\mathrm{Char}}
\newcommand{\M}{\mathcal M}
\newcommand{\F}{{GF}}

\begin{document}

\title{A note on full weight spectrum codes 
}

\author{Tim L. Alderson
}

\thanks{{\scriptsize
\hskip -0.4 true cm MSC(2010): Primary: 94B05  ; Secondary: 94B65 \and  94B25.
}}

\maketitle

\begin{center}
\end{center}

\begin{abstract}  A linear $ [n,k]_q $ code $ C $ is said to be a full weight spectrum (FWS) code if there exist codewords of each weight less than or equal to $ n $.  In this brief communication we determine necessary and sufficient conditions for the existence of linear $ [n,k]_q $ full weight spectrum (FWS) codes. Central to our approach is the geometric view of linear codes, whereby columns of a generator matrix correspond to points in $ PG(k-1,q) $.
\end{abstract}

\vskip 0.2 true cm


\pagestyle{myheadings}
\markboth{\rightline {\sl   \hskip 8.5 cm  T. L. Alderson }}
         {\leftline{\sl   \hskip 8.5 cm  T. L. Alderson}}

\bigskip
\bigskip





\section{\bf Introduction}
\label{intro}
\vskip 0.4 true cm

The weight spectra of codes have been investigated in many works over the years, and for different purposes. In \cite{MacWilliams1963}, MacWilliams exploited the relationship between  a code and its dual to determine the existence of a linear binary code with a given weight set. Delsarte \cite{Delsarte1973}, studied the cardinality of weight sets, utilizing them to compute the weight distributions of code cosets. Other discussions concerning the weight set and its cardinality can be found in \cite{Slepian1956,Enomoto1987}.\\
Recently,  Shi \textit{et. al.} \cite{Shi2018,Shi2018a}  studied a combinatorial problem concerning the maximum number $ L(k,q) $ of distinct weights  a  linear code of dimension $ k $ over $ GF(q) $ may realize.  Obviously  $ L(k,q)  \le \frac{q^k-1}{q-1}$. In \cite{Shi2018},  this bound is shown to be sharp for binary codes, and for all $ q $-ary codes of dimension $ k=2 $. We note that  in 2015, Haily and Harzalla \cite{Haily2015} also established the existence of binary codes meeting this bound.  Shi \textit{et. al.}  went on to conjecture that  the bound is sharp for all $ q $ and $ k $. This conjecture was proved correct in \cite{Alderson2018}.   Codes meeting this bound are called \textit{maximum weight spectrum} (MWS) codes.\\
A further refinement was also investigated in  \cite{Shi2018}
by the introduction of the function $ L(n, k, q) $, denoting the maximum number of non-zero weights an $ [n, k]_q $ code may have. They observed that an immediate upper bound is  $ L(n, k, q)\le n $. In this short communication we establish necessary and sufficient conditions for the existence of  codes meeting this bound. Such codes will be said to be \textit{full weight spectrum} (FWS) codes.\\


\section{\bf Preliminaries}
\vskip 0.4 true cm

\subsection{Linear codes and weight sets}

A linear $ [n,k]_q $ code is a $ k $ dimensional subspace of $ GF(q)^n $, where  $GF(q)$ denote the finite field with $q$ elements. Each of the subspace vectors are called codewords.  The (Hamming)\textit{ distance} between two codewords is the number of coordinates in which they differ, whereas the (Hamming) \textit{ weight} of a codeword $c\in \F(q)^n$ is the number of non-zero coordinates of $ c $.  The minimum distance, $ d $ of a linear code is the least pairwise distance among codewords, and is equal to the least non-zero codeword weight.   
\begin{defn}
	For an $[n,k]_q$ code $C$ we define the \emph{weight set of $C$} as
	$$w(C)=\left\{w(c) \mid c \in C\setminus \{0\}\right\}.$$
\end{defn}

Given an  $[n,k]_q$ \emph{code} $C$,  a $ k\times n $  matrix $ G $ is said to be a \emph{generator matrix}   for $C$ if the row vectors of $ G $ span $C$. \\

An $[n,k]_q$ code  $ C $ of dimension $ k\ge 2 $ is said to be   \emph{non-degenerate} if no coordinate position is identically zero. 
Throughout,  by $[n,k,d]_q$ code we will denote an $[n,k]_q$ code $C$ whose minimum distance is $d$. Moreover, unless specified otherwise, all codes discussed here are assumed to be non-degenerate.

\subsection{Projective systems}

We shall find the geometric view of linear codes to be of use. This view of linear codes is detailed in \cite{MR1186841}. Let us first give a brief 
overview of some fundamentals of finite projective geometries. For a detailed introduction we refer to the recent book by Ball \cite{Ball2015}. We let $PG(k,q)$ represent the finite projective geometry of dimension $k$ and order $q$.  A result of Veblen and Young
\cite{MR0179666} shows that all finite projective spaces of order $ q $, and of fixed dimension greater than two are isomorphic. The space $PG(k,q)$ can be modelled most easily with the vector space of dimension $k+1$ over the finite field $GF(q)$.  In this model, the one-dimensional subspaces represent the
points, two-dimensional subspaces represent lines, etc.  

Using this model, it is not hard to show
by elementary counting that the number of points of $PG(k,q)$ is given by \[\theta_q(k)=\frac{q^{k+1}-1}{q-1}.\]

A \textit{$d$-flat} $\Pi$ in $PG(k,q)$ is a subspace isomorphic to $PG(d,q)$; if $d=k-1$, the subspace $\Pi$ is called a \textit{hyperplane}.   A set of $ m $ points in $ PG(k,q) $ is said to be in \textit{general position} if every subset of $ k + 1 $ points form a basis for $ PG(k,q) $. We note at this time that through any proper non-empty subset $ S $, of a basis $ B $, there exists at least one hyperplane meeting $ B $ precisely in $ S $.\\ 
In finite geometry, a set of $ m $ points in general position  is called an $ m- $\textit{arc} (or an \textit{arc} of size $ m $). The literature on arcs is rich. An $ m $-arc in $ PG(k, q) $  is equivalent to a linear $ [m, k + 1]_q $ maximum distance separable (MDS)  code. MDS codes attain the Singleton bound, and therefore possess the strongest error correction possible. The classical example of an $ q+1 $-arc in $ PG(k,q) $ is the
set of points of the normal rational curve $ (1, t, t^2,\ldots, t^n)
, (t \in GF(q) \cup \{\infty\}) $. In the case $ k=2 $ all $ (q+1) $-arcs are conics. The $ (q+2)$-arcs  are known to exist in $ PG(2,q) $ when $ q $ is even, and are called hyperovals.

Central to the geometric view of linear codes is the idea of a projective system. 

\begin{defn}
	A \emph{projective $[n,k,d]_q$-system} is a finite (multi)set $\M$ of points 
	of  $PG(k-1,q)$, not all of which lie in a hyperplane, where $n=|\M|$ , and $$n-d=\max\{ |\M\cap H| \mid H \subset PG(k-1,q), \dim(H)=k-2\}.$$ Note that the cardinalities above are counted with multiplicities in the case of a multiset. We denote by $ m(P) $ the multiplicity of the point $ P $ in $ \M $.  
\end{defn}

Let $ C $  be an $ [n,k]_q $ code with $ k\times n $ generator matrix $ G $. Note that multiplying any column of $ G $ by a non-zero field element yields a generator matrix for a code which is equivalent to $ C $. Consider the (multi)set of one-dimensional subspaces of $ GF(q)^n $ spanned by the columns of $ G $. In this way the columns may be considered as a (multi)set $ \M $ of points of $ PG(k-1,q) $.  

For any non-zero vector $ v=(v_1,v_2,\ldots,v_k) $ in $ GF(q)^k $, it follows that the projective hyperplane 
\[
v_1x_1+v_2x_2+\cdots + v_kx_k=0
\]
contains $ |\M|-w $ points of $ \M $ if and only if the codeword $ vG $ has weight $ w $.  Thus,  linear (non-degenerate)  $ [n,k,d]_q $ codes and projective $ [n,k,d]_q $  systems are equivalent objects. That is to say, there exists a linear $ [n,k,d]_q $ code if and only if there exits a projective $ [n,k,d]_q $ system.

\begin{defn}
	Let $\M$ be a multiset in $ \Pi=PG(k-1,q) $. We define the \emph{character} function of $ \M $, denoted $ \ch_{\M} $ (or $ \ch $, if $ \M $ is clear from the context), mapping the power set of $ \Pi $ to the non-negative integers:
	\[
	\ch(A)=\sum_{P\in A} m(P).
	\]
	So $\ch (A) $ is  the number, including multiplicity, of points in $ \M \cap A$.  
\end{defn}

\section{\bf Full weight spectrum codes}
\vskip 0.4 true cm

Recently in \cite{Shi2018}, Shi \textit{et. al.}  introduced the function $ L(n, k, q) $, describing the maximum number of non-zero weights an $ [n, k]_q $ code may have. Clearly,  $ L(n, k, q)\le n $, motivating the following definition.

\begin{defn}
	An $ [n,k]_q $ code is said to be a Full Weight Spectrum (FWS) code, if the cardinality of it's weight set is $ n $, that is $ |W(C)|=n $.
\end{defn}

The following summarizes some of the main results  regarding $ L(n, k, q) $.

\begin{prop}\label{prop: Si summary}
	For all prime powers $ q $, and all positive integers $ k $:
	\begin{enumerate}
		\item $  L(n,k,q)\le L(n+1,k,q) $;
		\item $  L(n,k,q)\le L(n,k+1,q) $;
		\item $ L(n,k,q)\le L(n,k,q^m) $;
		\item $ \displaystyle L(n,k,q)\le L(k,q) = \frac{q^{k}-1}{q-1} $;
		\item $ \displaystyle \lim_{n\to \infty}L(n,k,q) \le L(k,q) = \frac{q^{k}-1}{q-1} (=\theta_q(k-1))  $;
		\item If $ n\ge \frac{q(q+1)}{2} $ then $ L(n,2,q)=q+1(=\theta_q(1)) $; \label{item: 2d}
		\item If $ k\ge 2 $, and $ L(n,k,q)=\theta_q(k-1) $, then  
		\[ n\ge  \left\lceil\frac{q\cdot\theta_q(k-1)}{2}\right\rceil
		= \left\lceil\frac{1}{2}\left[q^k+q^{k-1}+\cdots+q  \right]\right\rceil;\]
		\item If $ k\ge 3 $, and $ n\ge q^{\frac{k^2+k-4}{2}} $, then there exists an (possibly degenerate) MWS code, so  $ L(n,k,q)=\theta_q(k-1) $. \label{item: old MWS} 
	\end{enumerate}	
\end{prop}
\begin{proof}
	For 1-3, and the inequalities in 4, 5 see \cite{Shi2018}. For 6, 7, 8, and  the equalities in 4, 5, see \cite{Alderson2018}. 
\end{proof}

Before discussing  FWS codes, we offer the following result on MWS codes, in answer to a question raised in \cite{Alderson2018}. This generalizes to higher dimensions part \ref{item: 2d}, and for $ k>3 $ improves significantly part \ref{item: old MWS} of Proposition \ref{prop: Si summary}.   

\begin{prop}\label{prop: better bound MWS}
For each $ k\ge 2  $ there exists an MWS code of length  \[n=\theta_q(k-2)\cdot  {\theta_q(k-1)\choose 2} =\frac{q(q^k-1)(q^{k-1}-1)^2}{2(q-1)^3}\]

Consequently, for $ k\ge 2 $, and $ \displaystyle  n\ge \theta_q(k-2)\cdot  {\theta_q(k-1)\choose 2}\approx \frac{q^{3k-4}}{2} $, there exists an (possible degenerate) MWS code.
\end{prop}

\begin{proof}
Denote the hyperplanes of $ \pi=PG(k-1,q)=\{H_0,H_1,\ldots,H_{\theta_q(k-1)-1}\} $, and define the projective system $ \M $ as follows. For each point $P\in \pi$, let $ \ch_\M(P)=\sum_{P\in H_i} i $. Note that for $ k=2 $, the $ H_i $'s are equal or disjoint. For  $ k\ge 3 $, and  $ i\ne j $, $ H_i\cap H_j $ is a $ (k-3) $-flat. Consequently, for $ k\ge 2 $, and  $ 0\le t <\theta_q(k-1) $  we have 
\[ \ch_\M(H_t)=\theta_q(k-3) {\theta_q(k-1)\choose 2} +[\theta_q(k-2)-\theta_q(k-3)]t .\]
It follows that the corresponding linear code is MWS, and $ n= \theta_q(k-2)\sum\limits_{i=1}^{\theta_q(k-1)-1}i $.
\end{proof}

We now move on to our discussion of FWS codes. First, we provide a geometric view of $ [n,k]_q $ FWS codes.

\subsection{Geometric view of FWS codes}

Given an $[n,k,d]_q$ code $C$, we can consider the associated projective $[n,k,d]_q$-system $\M(C)$, whose points are given by the columns of the generator matrix. A direct consequence of the definitions is the following.

\begin{lem}
	Let $ C$ be an $[n,k]_q$ code over $\F (q)$, and let $\M:=\M( C)$ be an associated projective system. There exists a codeword of weight $ s $ in $ C $ if and only if there exists a hyperplane  $H$ in  $ \Pi $ with $\ch_\M(H)= n-s $.
\end{lem}

A natural consequence of the above Lemma is the following characterisation of FWS codes.

\begin{lem}\label{lem:GeometricConditions}
	Let $ \mathcal{H} $ denote the collection of all hyperplanes in $ PG(k-1,q) $. There exists an $[n,k]_q$ FWS code if and only if there exists an $ [n,k,d]_q $ projective system  $\M$  such that $ \ch :{\mathcal{H}}\to \{0,1,\ldots,n-1\} $ is a surjection.   
\end{lem}

\subsection{Existence of FWS codes}
In this section we investigate  the existence of  FWS codes. Since points and hyperplanes coincide in $ PG(1,q) $, the case of 2-dimensional FWS codes is  treated separately. 

\begin{lem}\label{lem: dimension 2}
	\begin{equation} \label{eqn: 2D}
	L(n,2,q)=\left\{\begin{array}{ll}
	\left\lfloor\frac{\sqrt{1+8n}-1}{2}\right\rfloor+1 & \textrm{, if }   n < \frac{q^2+q}{2}\\
	q+1 & \textrm{ otherwise.} 
	\end{array}\right.
	\end{equation} 
\end{lem} 
\begin{proof}
	Fix $ n $ and $ q $. An $ [n,2]_q $ code is determined by specifying the corresponding projective multiset $ \M $ in $ PG(1,q) =\{P_0,P_1,\ldots,P_{q}\}$. In this setting,  the number of (distinct) non-zero weights corresponds to the number of distinct point multiplicities $ m(P_i) $.  From the basic combinatorial theory, it follows that the number of distinct non-zero weights is maximized when the set of characters is $ \{0,1,2,3,\ldots,t\} $, where $ t$ is chosen maximally. Such a value $ t $ is the maximal integer solution to $ \frac{x(x+1)}{2}\le n $. The top inequality in (\ref{eqn: 2D}) follows since $ |W(C)|=t+1 $.  The bottom inequality follows from the cardinality of $ PG(1,q) $.
\end{proof}

\begin{lem} \label{lem:FWS 2D}
	An $ [n,2]_q $ FWS code exists if and only if $ n\le3$.
\end{lem}

\begin{proof}
Suppose $ C $ is an $ [n,2]_q $ FWS code, so that $ |W(C)|=n $, and consider the set $ \{m_1,m_2,\ldots,m_n\} $ of distinct multiplicities of the corresponding multiset in $ \ell=PG(1,q)=\{P_0,P_1,P_1,\ldots,P_q\} $. It holds that
\[
n\ge \sum\limits_{i=1}^{n}m_i  \ge \sum\limits_{i=0}^{n-1}i =\frac{n(n-1)}{2},
\]
giving $ n\le 3 $. 
For the other direction, observe that the projective systems $ \M_1=\{P_1,P_2,P_2\} $, and $ \M_2=\{P_1,P_2\} $ correspond to $ [n,2]_q $ FWS codes of length $ 3 $, and $ 2 $ respectively.
\end{proof}

\begin{lem} \label{lem:FWS if short}
	Let $ k\ge 3 $.	If $ n<2^{k} $ then $ L(n,k,q)=n $
\end{lem}
\begin{proof}
	Fix $ k $ points $ P_1,P_2,\ldots,P_{k} $ in general position (i.e. a basis) in  $ PG(k-1,q) $. Choose $ t $ maximal such that $2^{t}\le n$. Note that by assumption, $ t\le k-1 $. Define the projective multiset $ \M $ by assigning multiplicities according to  $m(P_i)=2^{i-1} $, $ 1\le i \le t $, and  $ m(P_{t+1})=n-(2^t-1)$. Every character from $ 0 $ to $ 2^{t}-1 $ is realized by at least one hyperplane meeting the first $ t $ points of the basis in at most $ t $ points. Further, every character from  $ n-(2^{t}-1) \;\; (\le 2^t) $ to $ n-1$ is realized by at least one hyperplane meeting $ P_{t+1} $ and at most $ t-1 $ of $ P_1,P_2,\ldots,P_t $. The result follows from Lemma \ref{lem:GeometricConditions}.  
	
\end{proof}

\begin{cor}
	$ L(n,k,2)=\left\{\begin{array}{ll}
	n & \textrm{, if }   n\le 2^{k}-1\\
	2^k-1 & \textrm{, otherwise.} 
	\end{array}\right. $
\end{cor}
\begin{proof}
	The first part follows from Lemma \ref{lem:FWS if short}. The second part follows from parts (1) and (4) of Proposition \ref{prop: Si summary}.
\end{proof}

\begin{lem} \label{lem: FWS no fat points}
	If $ C $ in an $ [n,k]_q $ FWS code, $ k\ge 3 $ with projective system $ \M $ in $ \Pi=PG(k-1,q) $, then each point $ P\in \Pi $ satisfies $ \ch(P)\le \left\lceil \frac{n}{2}\right\rceil $
\end{lem}

\begin{proof}
	Suppose by way of contradiction that there exists a point $ P\in \Pi $ with $ \ch(P)=t>\left\lceil \frac{n}{2}\right\rceil $. Any hyperplane $ H$ in $\Pi $ either meets $ P $, or does not. In the first case  $ \ch{(H)}\ge t >\left\lceil \frac{n}{2}\right\rceil $.  In the second case  	$ \ch{(H)}\le n-t < n-\left\lceil \frac{n}{2}\right\rceil \le\left\lceil \frac{n}{2}\right\rceil $. Consequently, there is no hyperplane with character $ \left\lceil \frac{n+1}{2}\right\rceil $. Since $ k\ge 3 $, $ C $ cannot be FWS (Lemma \ref{lem:GeometricConditions}).
\end{proof}

\begin{lem}\label{lem:FWS has plump point}
	If $ C $ is an $ [n,k]_q $ FWS code, $ n\ge 2^k $, $ k\ge 3 $ with projective system $ \M $ in $ \Pi=PG(k-1,q) $, then there exists a point $ P $ which satisfies $ \ch(P)\ge n-2^{k-1}+1 $.
\end{lem}
\begin{proof}
	Since $ C $ is FWS, there exist hyperplanes $ H_1,H_2,\ldots,H_{k-1} $ with $ \ch(H_i)=n-2^{i-1} $, for each $ 1\le i \le k-1 $. Observe that $ \ch(\mathcal{M}\setminus H_1)=1 $, and $ \ch(\mathcal{M}\setminus H_2)=2 $, so $ \ch(H_1\cap H_2)\ge n-3 $. Inductively we obtain $ \displaystyle \ch(\cap_{i=1}^{m}H_i)\ge n-2^m+1 $, $ 1\le m \le k-1 $.\\ Moreover, since $ \ch(H_m)=n-2^{m-1}<\ch(\cap_{i=1}^{m-1}H_i) $, $ 2\le m \le k-1 $ we must have $$ \dim(\cap_{i=1}^{m}H_i)\le \dim(\cap_{i=1}^{m-1}H_i)-1. $$ Further, since each $ H_i $ is an hyperplane of $ \Pi $, and $ m\le k-1 $, we must also have  $ \dim(\cap_{i=1}^{m}H_i)\ge \dim(\cap_{i=1}^{m-1}H_i)-1 $. Therefore, there exists a point $ P $ such that
	\[
	\{P\}=\cap_{i=1}^{k-1}H_i \text{ and } \ch(P)\ge n-2^{k-1}+1.
	\]
\end{proof}

\begin{cor}
	There exists an $ [n,k]_q $ FWS code if and only if $ n<2^{k} $
\end{cor}

\begin{proof}
For $ k=2 $ the result follows from Lemma \ref{lem:FWS 2D}, so we consider $ k\ge 3 $. Sufficiency follows from Lemma \ref{lem:FWS if short}.  For necessity,  suppose $ C $ is FWS with $ n\ge 2^{k} $. From Lemma \ref{lem:FWS has plump point}, there exists some point with $ \ch(P)\ge n-(2^{k-1}-1)>  \left\lceil \frac{n}{2}\right\rceil$ (since $ 2^{k-1}-1<\frac{n}{2} $). Lemma \ref{lem: FWS no fat points} gives a contradiction.
\end{proof}

\section{\bf Arcs and $ L(n,k,q) $}
\vskip 0.4 true cm

By assigning multiplicities to the points of an $ m $-arc in $ PG(k-1,q) $ we may construct codes having weight sets with a pleasant  combinatorial structure. For ease of presentation let us establish the following notation. For a given code length $ n$, let $ m=\lfloor \log_2(n)\rfloor $. Let $ r=n-2^m+1 $, and let  $ S_k=\{r+\alpha \mid 0\le \alpha <2^m, w(\alpha)\le k-2\} $. Here $ w(\alpha) $ is the weight of the binary representation of $ \alpha $.  

\begin{lem}\label{lem: bound 1 on L(n,k,q)}
	If  $ n < 2^{q+1} $, and $ k\ge 3 $ then $ L(n,k,q)\ge \sum\limits_{i=0}^{k-1} {m\choose i} +s $, where $  s=|\{a\in S_k\mid   w(a)\ge k \}| $. 
\end{lem}

\begin{proof}
	Let $ \Pi=PG(k-1,q) $, and let $ K=\{P_0,P_1,\ldots,P_{m}\} $ be a set of points in general position. Note that since  $ n<2^{q+1} $, such a set of points exists. Construct the projective multiset $ \M $ by assigning multiplicities as follows: $ m(P_i)=2^i $, $ 0\le i \le m-1 $, and $ m(P_m)=r=n-2^m+1 $. Through each (possibly empty) subset of $ k-1$ or fewer points from $ P_0,P_1,\ldots,P_{m-1} $ there is at least one hyperplane meeting $ K $ in exactly that subset.    If $ H $ is one of these hyperplanes, then the binary representation  of $ \ch(H) $ is an $ m $-bit string with weight at most $ k-1 $. Clearly no two of these hyperplanes have the same character.  Now consider the hyperplanes through $ P_m $. The set of characters of these lines is precisely the set $ S_k $.  The result follows from the definition of $ s $.   
\end{proof}

In the case that $ q $ is even, the existence of hyperovals in $ PG(2,q) $ gives a slight extension to the previous lemma.

\begin{lem}
	If  $ 2^{q+1} \le n < 2^{q+2} $, then $ L(n,3,q)\ge {q+1\choose 2} +1 +s $. In particular, $ L(2^{q+2}-1,3,q)\ge {q+2\choose 2} $.
\end{lem}

\begin{proof}
	The proof is entirely similar to the previous lemma, though we must take into account that each line in $ PG(2,q) $ meets an hyperoval in either $ 0 $, or $ 2 $ points. 
\end{proof}

There is much interest in determining the size of the largest arc in $ PG(k,q) $, $ k>2 $. According to the Main Conjecture on MDS Codes \cite{Ball2012b,Ball2012}, taking $ q>k $, such arcs have size bounded above by $ q+1 $ unless $ k=3 $ or $ k=q-1 $ and $ q $ is even, in which case $ n\le q+ 2 $. Employing a construction as in the previous two Lemmata we arrive at the following.

\begin{lem}
	If there exists an $ (q+2) $-arc in $ PG(k,q) $, $ k\ge 3 $ then there exists an $ [2^m-1,k]_q $ code where $|w(C)|=\sum\limits_{i=0}^{k-1} {q+2 \choose i} $
\end{lem}

\begin{rem}
	For fixed values of $ k $, and $ q $, the nature of the construction in the previous lemma makes it a relatively simple matter to enumerate codewords of each weight. In view of the main conjecture on MDS codes, it may therefore be of interest in future work to establish the non-existence of these codes.
\end{rem}    

\section{\bf Conclusions and open problems}
\vskip 0.4 true cm

In this brief communication we have determined necessary and sufficient conditions for the existence of full weight spectrum codes, \textit{i.e.} codes  achieving $ L(n,k,q)=n $.  Our results relate to the more general problem of determining $ L(n,k,q) $. Taking into account the result of \cite{Alderson2018}, and Proposition \ref{prop: better bound MWS}, we have 
\[
L(n,k,q)=\left\{ \begin{array}{cl}
n & \textrm{ iff } n< 2^k \\
\theta_q(k-1) & \textrm{ if } n\ge \min \{q^{\frac{k^2+k-4}{2}}, \frac{q(q^k-1)(q^{k-1}-1)^2}{2(q-1)^3} \}
\end{array}\right. .
\]
It may therefore be of interest to establish values or general bounds for  $ L(n,k,q)$ where $ 2^k-1\le n < \min \{q^{\frac{k^2+k-4}{2}}, \frac{q(q^k-1)(q^{k-1}-1)^2}{2(q-1)^3} \} $.\\

\vskip 0.4 true cm

\section*{\bf Acknowledgement}
The author acknowledges the support of the Natural Sciences and Engineering Research Council of Canada (NSERC).

\vskip 0.4 true cm


\bibliographystyle{plain}
\bibliography{AldersonFWS}
%
%
%
%
%

\bigskip
\bigskip

{\footnotesize \pn{\bf Tim L. Alderson}\; \\ {Department of
Mathematics and Statistics}, {University
of New Brunswick Saint John, 100 Tucker Park Rd, P.O.Box 5050,} {Saint John, New Brunswick, Canada}\\
{\tt Email: Tim@unb.ca}\\

\end{document}